\RequirePackage[2014/01/01]{latexrelease}   
\documentclass[hidelinks,11pt,english]{article} 
\usepackage{hyperref}
\usepackage{times}
\usepackage{helvet}
\usepackage{courier}
\usepackage{amsfonts}
\usepackage{amsmath}
\usepackage{etex,etoolbox}

\usepackage{booktabs}
\usepackage{amsthm}
\usepackage{graphicx}
\graphicspath{ {./images/} }
\usepackage{mathrsfs}
\usepackage{enumerate}
\usepackage{color}
\usepackage{etex,etoolbox}
\usepackage{float}
\usepackage{xcolor}
\usepackage[round]{natbib}
\usepackage[margin=1in,bottom=0.9in,includefoot]{geometry}
\usepackage[toc,page]{appendix}

\newcommand{\musiclab}{\textsc{MusicLab}}

\makeatletter
\providecommand{\@fourthoffour}[4]{#4}
\def\fixstatement#1{%
	\AtEndEnvironment{#1}{%
		\xdef\pat@label{\expandafter\expandafter\expandafter
			\@fourthoffour\csname#1\endcsname\space\@currentlabel}}}

\globtoksblk\prooftoks{1000}
\newcounter{proofcount}

\long\def\proofatend#1\endproofatend{%
	\edef\next{\noexpand\begin{proof}[Proof of \pat@label]}%
		\toks\numexpr\prooftoks+\value{proofcount}\relax=\expandafter{\next#1\end{proof}}
	\stepcounter{proofcount}}

\def\printproofs{%
	\count@=\z@
	\loop
	\the\toks\numexpr\prooftoks+\count@\relax
	\ifnum\count@<\value{proofcount}%
	\advance\count@\@ne
	\repeat}
\makeatother

\DeclareMathOperator*{\argmax}{arg-max}

\newtheorem{theorem}{Theorem}[section]
\newtheorem{thm}{Theorem}[section]

\newtheorem{proposition}[theorem]{Proposition}

\newtheorem{cor}[theorem]{Corollary}
\fixstatement{thm}
\fixstatement{lem}
\fixstatement{thm}
\fixstatement{cor}
\newtheorem{example}{Example}[section]

\begin{document}
	
	\title{Trial-Offer Markets with Continuation}
\author{
	Pascal Van Hentenryck\footnote{Industrial and Operations Engineering \& Computer Science and Engineering, University of Michigan, Ann Arbor (pvanhent@umich.edu).}
	\and
	Alvaro Flores\footnote{College of Engineering \& Computer Science, Australian National University.}
	\and
	Gerardo Berbeglia\footnote{Melbourne Business School, The University of Melbourne.}
}
	
	\maketitle
	\begin{abstract}

			Trial-offer markets, where customers can sample a product before
			deciding whether to buy it, are ubiquitous in the online
			experience. Their static and dynamic properties are often studied by
			assuming that consumers follow a multinomial logit model and try
			exactly one product. In this paper, we study how to generalize    
			existing results to a more realistic setting where consumers can try
			multiple products. We show that a multinomial logit model with
			continuation can be reduced to a standard multinomial logit model
			with different appeal and product qualities. We examine the
			consequences of this reduction on the performance and predictability
			of the market, the role of social influence, and the ranking
			policies.
	\end{abstract}

\section{Introduction}

With the ubiquity of online market places such as {\tt Amazon} and
{\tt iTunes}, there has been increasing interests in understanding and
modeling the behavior of such trial-offer markets, where customers
sample a product before deciding whether to buy it. These online
markets are particularly interesting because of their greater
opportunities in shaping the customer experience and their flexibility
in exploiting visibility bias \cite{Lerman2014,buscher2009you,maille2012sponsored,joachims2005accurately} and social signals
\cite{engstrom2014demand,viglia2014please}.

Traditionally, such markets have been studied using extensions
of Multinomial Logit Models
\cite{krumme2012quantifying,Lerman2014,PLOSONESI,QRANKING,daly1978improved,talluri2004revenue,rusmevichientong2010assortment,rusmevichientong2010dynamic}, where
participants can only sample one product (e.g., listening to a song)   
before deciding whether to buy the product. Multinomial Logit Models
have been studied in marketing for several decades \cite{luce1959},  but
the addition of social signals has been shown to affect the market
behavior significantly      
\citep{salganik2006experimental,Lerman2014,PLOSONESI}. Hence it is   
interesting to study how these recent results are affected by more
realistic settings in which consumers can try multiple products. Also, experimental analysis 
using eye-tracking inspired cascade models, introduced first by \cite{Craswell_2008}, and showed fit
experimental data better than separable models in presence of position bias. The intuition behind this model, is that users consider products in a top to bottom fashion, as they were presented in an ordered list, and they only look to the next product if the current one was not selected. 

This paper is an attempt at generalizing Multinomial Logit Models to
account for a richer class of customer behavior. It endows the
Multinomial Logit Model with a notion of continuation, which enables
participants to sample multiple products before making a purchase. The
paper studies how this generalization affects market efficiency and
the role of social influence. The main contributions of the paper can
be summarized as follows:
\begin{enumerate}
	\item We show that a trial-offer market with continuation can be
	reduced to a traditional trial-offer market by adjusting the quality
	and appeal of the products and we quantify how the continuation
	model affects market efficiency;
	
	\item We show that, under a natural continuation model, the
	quality-ranking policy, where the products are ranked by quality, is
	preserved by the reduction, but not the performance ranking, which
	optimizes the market performance at each step. We also show that
	social influence remains beneficial in this setting under the
	quality ranking;                                             
	
	\item Finally, we show experimental results that indicate that the
	popularity ranking, which ranks the product by popularity, benefits
	more from the generalization than the quality and performance
	ranking, unless the continuation is strongly dependent of the
	product just sampled. This improvement however is not enough to
	bridge the gap with the performance and quality rankings.
\end{enumerate}

\section{Trial-Offer Markets}
\label{section:market}

This paper considers trial-offer markets in which participants can try
a product before deciding whether to buy it. Such settings are common
in online cultural markets (e.g., books, songs, and videos). In this
paper, the trial-offer market is composed of $n$ products and each
product $i \in \{1,\ldots,n\}$ is characterized by two values:
\begin{enumerate}
	\item Its {\em appeal} $A_i$ representing the inherent preference
	of trying product $i$;
	
	\item Its {\em quality} $q_i$ representing the
	probability of purchasing product $i$ given that it was tried.
\end{enumerate}
Each participant, when entering the market, is presented with a
product list $\pi$: She then tries a product $s$ in $\pi$ and decides
whether to purchase $s$ with a certain probability. The product list
is a permutation of $\{1,\ldots,n\}$ and each position $p$ in the list
is characterized by its {\em visibility} $v_p > 0$ which is the
inherent probability of trying a product in position $p$. Since the
list $\pi$ is a bijection from positions to products, its inverse is
well-defined and is called a ranking. We denote rankings by $\sigma$
in the following, $\pi_i$ denotes the product in position $i$ of the
list $\pi$, and $\sigma_i$ denotes the {\em position} of product $i$
in the ranking $\sigma$. Therefore $v_{\sigma_i}$ denotes the
visibility of the position of product $i$.

The probability of trying product $i$ given a list $\sigma$ is
\[
p_{i}(\sigma) =  \frac{v_{\sigma_i} A_i}{\sum_{j=1}^n v_{\sigma_j} A_j}.
\]
Given a ranking $\sigma$, the expected number of purchases is
\begin{align}\label{exp_down}
\lambda(\sigma)=\sum_{i=1}^n p_{i}(\sigma) \ q_i.
\end{align}
The traditional static market optimization problem consists of finding
a ranking $\sigma^*$ maximizing $\lambda(\sigma)$, i.e.,
\begin{align}\label{p-ranking_formulae}
\sigma^*=\argmax_{\sigma \in S_n} \sum_{i=1}^n p_{i}(\sigma) \ q_i
\end{align}
\noindent
where $S_n$ represents the symmetry group over
$\{1,\ldots,n\}$. Observe that consumer choice preferences for trying
the products are essentially modeled as a discrete choice model based
on a multinomial logit \cite{luce1959} in which product utilities are
affected by their position.

\paragraph{Social Influence}

Following \cite{krumme2012quantifying}, this paper considers a dynamic
market where the appeal of each product changes over time according to
a social influence signal. Given a social signal $d=(d_1,\ldots,d_n)$,
where $d_i$ denotes the number of purchases of product $i$, the appeal
of $i$ becomes $A_i + d_i$ and hence the probability of trying $i$       
given a list $\sigma$ becomes
\[
p_{i}(\sigma,d) =  \frac{v_{\sigma_i} (A_i + d_i)}{\sum_{j=1}^n v_{\sigma_j} (A_j + d_j)}.
\]
Note that the probability of trying a product depends on its position
in the list, its appeal, and its number of purchases ($d_{i,t}$) at
time $t$. As the market evolves over time, the number of purchases
could dominate the appeal, and the sampling probability of a product
becomes its market share. Without social influence, a dynamic market
reduces to solving the static optimization problem repeatedly. This
set-up is the {\em independent condition}.

In the following, without loss of generality, we assume that the
qualities and visibilities are non-increasing, i.e., $ q_{1}\geq
q_{2}\geq\cdots\geq q_{n} $ and $ v_{1}\geq v_{2}\geq\cdots\geq
v_{n}$.  We also assume that the qualities and visibilities are
known. In practical situations, the product qualities are obviously
not known. But, as shown by \citet{PLOSONESI}, they can be recovered
accurately and quickly, either before or during the market execution.
For simplicity, we use $ a_{i,t} = A_i + d_{i,t} $ to denote the
appeal of product $i$ at step $t$. When the step $t$ is not relevant,
we omit it and use $a_i$ instead.

\paragraph{Ranking policies}

Following \cite{PLOSONESI}, this paper explores several ranking     
policies. The {\em performance ranking} maximizes the expected number
of purchases at each iteration, exploiting all the available
information globally, i.e., the appeal, the visibility, the purchases,
and the quality of the products. More precisely, the performance
ranking at step $k$ produces a ranking $\sigma_k^*$ defined as
\[
\sigma_k^* = \argmax_{\sigma\in S_n} \sum_{i=1}^n p_{i}(\sigma,d_k)\cdot q_i
\]
where $d_k = (d_{1,k},\ldots,d_{n,k})$ is the social influence signal
at step $k$.  The performance ranking uses the probability
$p_{i}(\sigma,d_k)$ of trying products $i$ at iteration $k$ given
ranking $\sigma$, as well as the quality $q_i$ of product $i$. The
performance ranking can be computed in strongly polynomial time and
the resulting policy is scalable to large markets
\cite{PLOSONESI}. The {\em quality ranking} simply orders the products
by quality, assigning the product of highest quality to the most
visible position and so on. With the above assumptions, a quality
ranking $\sigma$ satisfies $\sigma_i = i \;\; (1 \leq i \leq n)$.  The
{\em popularity ranking} was used by \citet{salganik2006experimental}
to show the unpredictability caused by social influence in cultural
markets. At iteration $k$, the popularity ranking orders the products
by the number of purchases $d_{i,k}$, but these purchases do not
necessarily reflect the inherent quality of the products, since they
depend on how many times the products were tried. We also follow
\cite{PLOSONESI} and use {\sc Q-rank}, {\sc D-rank}, and {\sc P-rank}
to denote the policies using the quality, popularity, and performance
rankings respectively. We also use {\sc R-rank} to denote the    
policy that simply presents a random order at each period.

\section{Trial-Offer Markets With Continuation}
\label{section:Continuation}

\begin{figure}[t]
	\hspace{1cm}
	\includegraphics[width=0.8\textwidth]{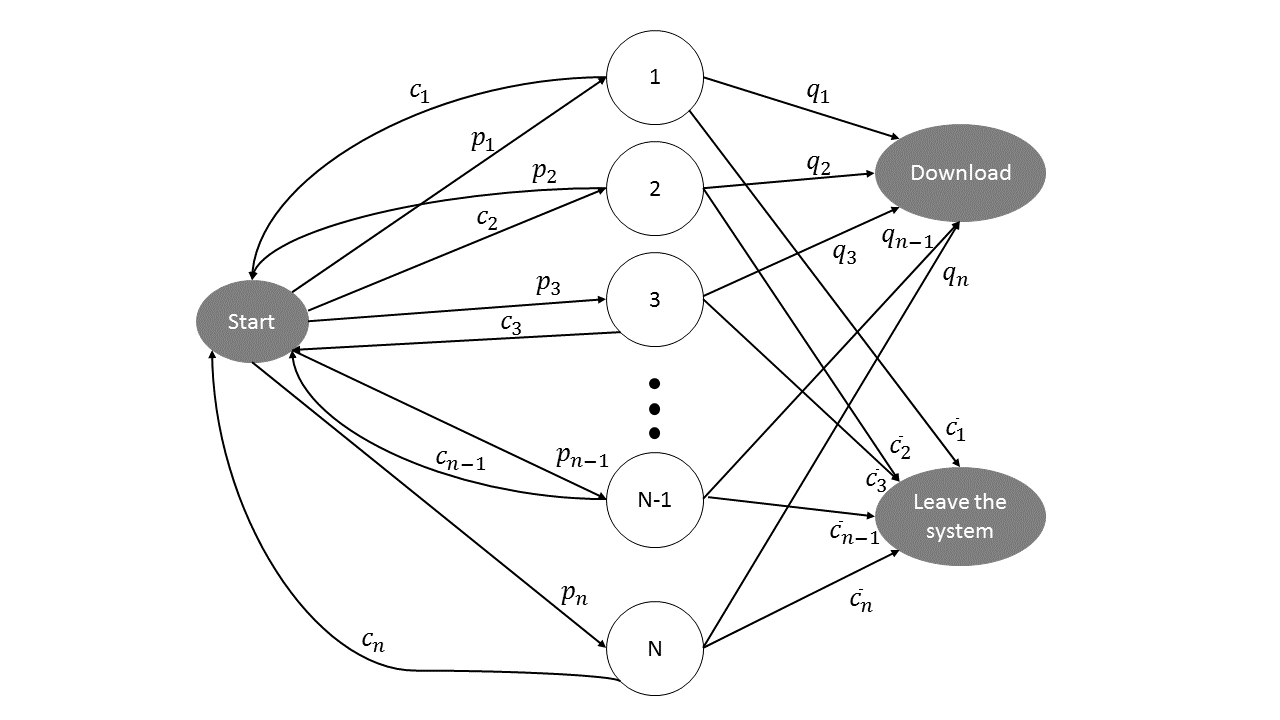}
	\caption{A Trial-Offer Market with Continuation.}
	\label{fig:journey}
\end{figure}

The main goal of this paper is to study trial-offer markets with
continuation, i.e., a setting where market participants can continue
shopping even when they decline to purchase the product just sampled.
We model such a trial-offer market by adding a continuation 
probability
\begin{equation}\label{cont_p}
c_{i}=f(\cdot)(1-q_{i})
\end{equation}
to continue shopping after a participant has declined to purchase
product $i$. In the above probability, the $(1-q_i)$ term represents
the fact that the participant has declined to purchase product $i$ and
the $f(\cdot)$ term represents a function that might depend on the
product quality, the current position, or even on another overall
measure (or a combination of all these factors). Figure
\ref{fig:journey} shows a graphic representation of a trial-offer
market with continuation. It uses $\overline{c}_{i}=1-c_{i}$ to denote
the probability that a participant leaves the market place after
sampling product $i$.

The expected number of purchases in the static version of the
trial-offer market with continuation for a ranking $\sigma$ is denoted
by $\overline{\lambda(\sigma)}$ and defined by

\begin{equation}\label{lambda_j}
\overline{\lambda(\sigma)}=\sum_{i=1}^{n}p_{i}(\sigma)(q_{i}+c_{i}\overline{\lambda(\sigma)})
\end{equation}

\noindent
Our primary objective is to maximize market efficiency, i.e., the
expected purchases:
\begin{equation}\label{cascade-obj}
\sigma^* = \argmax_{\sigma \in S_n} \overline{\lambda(\sigma)}.
\end{equation}
Note that the higher this objective is, the lower the probability that
consumers try a product but then decide not to purchase it. Hence, if
we interpret this last action as an inefficiency, maximising the
expected efficiency of the market minimizes unproductive trials.

\begin{figure}[t]
	\begin{centering}
		\includegraphics[width=0.6\linewidth]{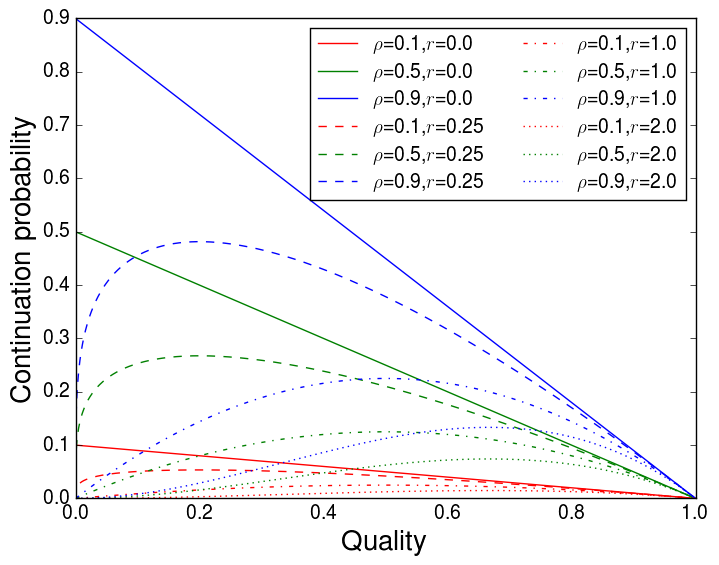}
	\end{centering}
	\centering{}
	\caption{Examples of the continuation probabilities for different values of $\rho$ and $r$; the $r$ parameter defines where the peak is (the maximum is always attained at $q=\frac{r}{r+1}$), and $\rho$ modulates how strong the continuation is.}
	\label{im:cont_prob}
\end{figure}

This paper proves a number of results when the continuation $c_i$ 
depends polynomially on $q_i$, i.e.,
\begin{equation}\label{ci_f}
c_{i}=\rho q_{i}^r(1-q_{i})
\end{equation}
where $\rho \leq 1$ controls the overall tendency to continuation and
$r\geq 0$ represents the influence of $q$. This choice is justified
intuitively by the fact that a market participant is more likely to
continue sampling if the product she tried is of high quality, because
it reflects on how good the other products potentially are.  Figure
\ref{im:cont_prob} depicts various choices of $\rho$ and $r$.

\section{Reduction to the Trial-Offer Model}
\label{section:relation}

This section shows that the trial-offer market with continuation can
be reduced to a trial-offer market. Indeed, rearranging the terms in
Equation \ref{lambda_j} leads to
\begin{align}\label{closed_form_deduction}
\overline{\lambda(\sigma)}&=\sum_{i=1}^{n}p_{i}(\sigma)(q_{i}+c_{i}\overline{\lambda(\sigma)})\nonumber\\
\overline{\lambda(\sigma)}&=\frac{\sum_{i=1}^{n}p_{i}(\sigma)q_{i}}{1-\sum_{i=1}^{n}p_{i}(\sigma)c_{i}} \nonumber
\end{align}

\noindent
Defining
\begin{equation}
\label{eq:pibar}
\overline{p_{i}(\sigma)}=\frac{p_{i}(\sigma)}{1-\sum_{i=1}^{n}p_{i}(\sigma)c_{i}} 
\end{equation}
we obtain 
\begin{align*}
& \overline{\lambda(\sigma)}=\sum_{i=1}^{n}\overline{p_{i}(\sigma)} \ q_{i}
\end{align*}
By definition of $p_{i}(\sigma)$, we have
\begin{align*}\label{p_cont1}
\overline{p_{i}(\sigma)}&=\frac{v_{\sigma_{i}}a_{i}}{\sum_{i=1}^{n}v_{\sigma_{i}}a_{i}}\cdot \frac{1}{1-\sum_{i=1}^{n}(c_{i}\cdot \frac{v_{\sigma_{i}}a_{i}}{\sum_{i=1}^{n}v_{\sigma_{i}}a_{i}})}\\
\overline{p_{i}(\sigma)}&=\frac{v_{\sigma_{i}}a_{i}}{\sum_{i=1}^{n}(1-c_{i})v_{\sigma_{i}}a_{i}}
\end{align*}
Now, by defining $\overline{a_{i}}=a_{i}(1-c_{i})$ and $\overline{q_{i}}=\frac{q_{i}}{(1-c_{i})}$,
we obtain 
\begin{align}
&\overline{\lambda(\sigma)}=\sum_{i=1}^{n}\frac{v_{\sigma_{i}}\overline{a_{i}} \; \overline{q_{i}}}{\sum_{i=1}^{n}v_{\sigma_{i}}\overline{a_{i}}}\nonumber 
\end{align}

\noindent
We have proven the following theorem:
\begin{thm}
	\label{reduction_thm_cont}
	A trial-offer market with continuation can be reduced to a trial-offer
	market by using the product qualities $\overline{q_{i}}$ and appeals
	$\overline{a_{i}}$ defined as follows:
	\begin{align}
	&\overline{q_{i}}=\frac{q_i}{1-c_i}\nonumber \\
	&\overline{a_{i}}=a_i(1-c_i). \nonumber
	\end{align}
\end{thm}

\begin{figure}[t]
	\begin{centering}
		\includegraphics[width=0.6\linewidth]{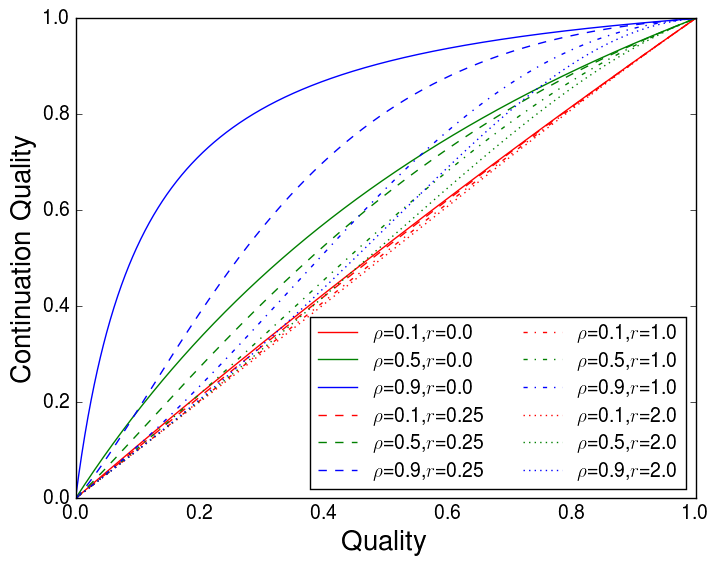}
	\end{centering}
	\centering{}
	\caption{Continuation qualities for different
		values of the $\rho$ and $r$ parameters; the lager $\rho$ is, the more {\em concave} the continuation
		quality becomes.}
	\label{im:ass_q}
\end{figure}

\noindent
In the following, $\overline{q_{i}}$ and $\overline{a_{i}}$ are called
the {\em continuation qualities} and {\em continuation appeals}, and  
Figure \ref{im:ass_q} depicts the continuation quality for different
values of $\rho$ and $r$. Observe how the continuation model typically
boosts the quality of the products, sometimes substantially. To understand this reduction intuitively,
we can rewrite Equation \ref{eq:pibar} as:
\begin{align*}
\overline{p_{i}(\sigma)}&=p_{i}(\sigma)\cdot\ \sum_{j=1}^{\infty}(\sum_{i=1}^{n}p_{i}(\sigma)c_{i})^j \\
\end{align*}
The value $\overline{p_{i}(\sigma)}$ can thus be interpreted as the
probability of sampling product $i$ in any number of steps. The
rewriting uses the fact that $\sum_{i=1}^{n}p_{i}(\sigma)c_{i}<1$ to
obtain an infinite sum and the term
$(\sum_{i=1}^{n}p_{i}(\sigma)c_{i})^j$ captures all the possible ways
to sampling $i$ in $j$ steps.

\section{Properties of the Market}
\label{section:performance}

\paragraph{Market Efficiency:}The first result links the expected
number of purchases of the market with and without continuation under
the performance ranking.

\begin{thm}\label{bounds_thm}
	Let $\pi_{c}^{*}$ and $\pi^{*}$ be optimal permutations for the
	trial-offer markets with and without continuation. Then,
	\begin{equation}
	\lambda(\pi^{*}) \leq \overline{\lambda}(\pi_{c}^{*})\leq \frac{\lambda(\pi^{*})}{1-\max\limits_{i}c_{i}}. \nonumber
	\end{equation}
\end{thm}
\begin{proof}
	The lower bound can be derived as follows:
	\begin{align*}
	\lambda(\pi^{*})&= \frac{\sum_{i=1}^{n}v_{i}a_{\pi^{*}_{i}}q_{\pi^{*}_{i}}}{\sum_{i=1}^{n}v_{i}a_{\pi^{*}_{i}}}\\
	\lambda(\pi^{*})&=\frac{\sum_{i=1}^{n}v_{i}a_{\pi^{*}_{i}}q_{\pi^{*}_{i}}}{\sum_{i=1}^{n}v_{i}a_{\pi^{*}_{i}}} \cdot \frac{1-\sum_{i=1}^{n}p_{i}(\pi^{*})c_{\pi^{*}_{i}}}{1-\sum_{i=1}^{n}p_{i}(\pi^{*})c_{\pi^{*}_{i}}}\\
	\lambda(\pi^{*})&=\underbrace{\frac{\sum_{i=1}^{n}p_{i}(\pi^{*})q_{\pi^{*}_{i}}}{1-\sum_{i=1}^{n}p_{i}(\pi^{*})c_{\pi^{*}_{i}}}}_{\overline{\lambda}(\pi^{*})}\underbrace{(1-\sum_{i=1}^{n}p_{i}(\pi^{*})c_{\pi^{*}_{i}})}_{\leq 1}\\
	\lambda(\pi^{*})&\leq\overline{\lambda}(\pi^{*}) \leq\overline{\lambda}(\pi_{c}^{*})
	\end{align*}
	where the last inequality holds because of optimality of $\pi_{c}^{*}$
	for $\overline{\lambda}(\cdot)$. The upper bound follows from 
	\[
	\overline{\lambda}(\pi_{c}^{*}) \leq \lambda(\pi_{c}^{*})\cdot \frac{1}{1-\sum_{i=1}^{n}p_{i}c_{\pi^{*}_{i}}} \leq \lambda(\pi^{*})\cdot \frac{1}{1-\max\limits_{i}c_{i}}.
	\]
\end{proof}

\noindent 
The following corollary considers the case where continuations depend polynomially on qualities (a proof is provided in the \nameref{appendix}). 
\begin{cor}\label{Opt_Bounds}
	Assume that $c_i=\rho q_i^r (1-q_i)$. It follows that 
	\begin{equation}\label{bound_opts_propq}
	\lambda(\pi^{*}) \leq \overline{\lambda}(\pi_{c}^{*})\leq \lambda(\pi^{*})\frac{1}{1-\frac{\rho r^r}{(r+1)^{r+1}}}
	\end{equation}
\end{cor}
\proofatend 	
The proof follows from examining the maximum value for the $c_i$. 
\begin{equation*}\label{bound_opts_proof_pow}
\max\limits_{i}c_{i} = \max\limits_{i}\rho q_{i}^r(1-q_{i}) \leq \rho \max\limits_{x \in [0,1]}x^r(1-x) = \frac{\rho r^r}{(r+1)^{r+1}}
\end{equation*}
where the last equality holds because the maximum value of $x^r(1-x)$
is reached when $x=\frac{r}{r+1}$.
\endproofatend

\noindent
When $\rho= 1$ and $r=1$, $\lambda(\pi_{c}^{*}) \leq
\frac{4}{3}\cdot\lambda(\pi^{*}) $ indicating a market that is at most
33\% more efficient.

Prior work on trial-offer markets with the social influence signal
considered here has shown that the quality ranking is asymptotically
optimal \cite{QRANKING}: The market converges towards a monopoly for
the product of highest quality. We now show that, when the
continuations are polynomial in product qualities, the quality ranking
is preserved by the reduction and hence the two markets, with and
without continuation, converge to the same equilibrium in market
shares.

\begin{proposition}\label{monotony_Poly}
	Let $c_{i}=\rho q_{i}^r(1-q_{i})$ with $\rho \in (0,1)$ and
	$r\geq0$. Then $q_i\leq q_j \Leftrightarrow \overline{q_i}\leq
	\overline{q_j}$.
\end{proposition}
\begin{proof}
	It is sufficient to show that $\overline{q_i}$, when viewed as a
	function of $q_i$, is increasing in $(0,1)$. Consider such function
	\[h(x)=\frac{x}{1-\rho x^r(1-x)}\]
	and its derivative
	\[
	\frac{dh(x)}{dx}=\frac{\rho x^{r}\left[(r-1)-rx\right]+1}{((1-\rho x^r(1-x))^2}.
	\]
	The denominator is greater than zero, so it remains to show that the
	numerator also is. The term $\rho x^{r}$ is increasing in $x$ and the
	term $\left[(r-1)-rx\right]$ is a line decreasing in $x$.  The product
	is minimized when $x=1$, in which case the product has a value of
	$-\rho$. Since $\rho \in (0,1)$, the minimum value of the numerator is
	$1 -\rho \geq 0$, which concludes the proof.
\end{proof}

\noindent
More importantly, it is also possible to show that, under the quality
ranking, the probability that the next purchase is product $i$ is the
same in the markets with and without continuation. Hence, from a
product standpoint, the markets behave very similarly.

\begin{proposition}
	\label{p_i}
	The probability $p_i$ that the next purchase (after any number of
	steps) is product $i$ is
	\begin{equation*}
	p_i=\frac{v_i a_i q_i}{\sum\limits_{j=1}^n v_j a_j q_j}.
	\end{equation*}
\end{proposition}
\proofatend
The probability that product $i$ is purchased in the first step is 
\begin{equation*}
p_i^{1st}= \frac{v_i a_i}{\sum\limits_{j=1}^n v_j \overline{a_j}} q_i,
\end{equation*}
More generally, the probability that product $i$ is purchased in step
$m$ while no product was purchased in earlier steps is:
\begin{equation*}
\label{noarreg}
p_i^{mth}=\left( \frac{\sum\limits_{j=1}^n v_j \overline{a_j} (1-\overline{q_j})}{\sum\limits_{j=1}^n v_j \overline{a_j}}\right)^{m-1} \frac{v_i a_i}{\sum\limits_{j=1}^n v_j \overline{a_j}} q_i.
\end{equation*}
Defining $\beta=(\sum\limits_{j=1}^n v_j a_j q_j)/(\sum\limits_{j=1}^n v_j \overline{a_j})$, Equation \ref{noarreg} becomes
\begin{equation*}
p_i^{mth}=\bigg(1-\beta\bigg)^{m-1}  \frac{v_i a_i}{\sum\limits_{j=1}^nv_j \overline{a_j}} q_i.
\end{equation*}
Hence the probability that the next purchased product is product $i$ is given by
\begin{equation*}
p_i=  \sum\limits_{m=0}^\infty\bigg(1-\beta\bigg)^m \frac{v_i a_i}{\sum\limits_{j=1}^n v_j \overline{a_j}} q_i.
\end{equation*}
Given the fact that $\beta<1$ we have:
\begin{equation*}
\sum\limits_{m=0}^\infty\bigg(1-\beta\bigg)^m=\frac{1}{\beta},
\end{equation*}
the probability that the next purchase is product $i$ is given by
\begin{align*}
p_i&=\frac{\sum\limits_{j=1}^n v_j \overline{a_j}}{\sum\limits_{j=1}^n v_j a_j q_j}\cdot \frac{v_i a_i}{\sum\limits_{j=1}^n v_j \overline{a_j}} q_i.\\
p_i&=\frac{v_ia_iq_i}{\sum\limits_{j=1}^n v_ja_jq_j}
\end{align*}
\endproofatend

A proof of Propositon \ref{p_i} is given in the \nameref{appendix}). In contrast, the same results do not hold for the performance ranking,
which may change when a continuation is used, as shown by the
following example.

\begin{example}
	Consider the following instance with 3 songs:
	\begin{itemize}
		\item Visibilities: $v_1=0.8$, $v_2=0.5$ and $v_3=0.1$
		\item Qualities: $q_1=0.9$, $q_2=0.2$ and $q_3=0.6$
		\item Appeals: $a_1=0.9$, $a_2=0.1$ and $a_3=0.3$
		\item Continuation parameters: $\rho=0.8$ and $r=0.7$
	\end{itemize}
	In this case, the performance ranking for the market without
	continuation is $\sigma^*=[1,2,3]$; It is $\sigma^*_c=[1,3,2]$
	for the continuation
	model. 
\end{example}

\paragraph{Position Bias:}This result generalizes the result shown in 
\cite{QRANKING}, to the continuation setting, and it means that we can always benefit 
from position bias. The formalization of this claim can be seen below (a proof is provided in the \nameref{appendix}))

\begin{thm}
	\label{thm:q-position-bias}
	Position bias increases the expected number of purchases under the
	quality-ranking policy, i.e., for all visibilities $v_i$, appeals
	$a_i$, qualities $q_i$ $(1 \leq i \leq n)$ and continuation probabilities $c_i$.This is,  after we make the reduction to the Associated Multinomial Logit, we have:
	\[
	\frac{\sum_{i}v_{i}\overline{a_{i}} \ \overline{q_{i}}}{\sum_{j}v_{j}\overline{a_{j}}}\geq\frac{\sum_{i}\overline{a_{i}}\ \overline{q_{i}}}{\sum_{j}\overline{a_{j}}}.
	\]
\end{thm}

\proofatend
Let $\overline{\lambda} = \frac{\sum_{i}v_{i}\overline{a_{i}}\overline{q_{i}}}{\sum_{j}v_{j}\overline{a_{j}}}$
be the expected number of purchases for the quality ranking. We have
\[
\sum_{i}v_{i}\overline{a_{i}}\left(\overline{q_{i}}-\overline{\lambda}\right)=0.
\]
Consider the index $k$ such that $\left(\overline{q_{k}}-\overline{\lambda}\right)\geq0$ and
$\left(\overline{q_{k+1}}-\lambda\right)<0$. Since $v_{1} \geq \ldots \geq v_n$, we have
\[
\sum_{i=1}^{k}v_{k}\overline{a_{i}}\left(\overline{q_{i}}-\overline{\lambda}\right)+\sum_{i=k+1}^{n}v_{k}\overline{a_{i}}\left(\overline{q_{i}}-\overline{\lambda}\right) \leq \sum_{i}v_{i}\overline{a_{i}}\left(\overline{q_{i}}-\overline{\lambda}\right) = 0
\]
and, given the fact that $v_k\geq 0$,
\[
\sum_{i=1}^{n}\overline{a_{i}}\left(\overline{q_{i}}-\overline{\lambda}\right) \leq 0.
\]
We had the desired result:
$
\lambda \geq \frac{\sum_{i=1}^{n}\overline{a_{i}}\overline{q_{i}}}{\sum_{i=1}^{n}\overline{a_{i}}}.
$
\endproofatend

\paragraph{Social Influence:}The last result in this section shows
that the social influence signals always benefit trial-offer markets
with continuation. The result is independent of the structure of
the continuation probabilities. The proof is a generalization 
of the result in \cite{QRANKING}. 

\begin{thm}
	The expected marginal rate of purchases is non-decreasing over time for the
	quality ranking under social influence in trial-offer markets with
	continuation.
\end{thm}

\begin{proof}
	Let
	\begin{equation}\label{lambdaproof}
	\mathbb{E}[D_{t}]=\frac{\sum_{i}v_{i}\overline{a_{i}}\overline{q_{i}}}{\sum_{i}v_{i}\overline{a_{i}}}=\overline{\lambda}
	\end{equation}
	
	\noindent	
	be the expected number of purchases at time $t$. The expected number
	of purchases at time $t+1$ conditional to time $t$ is
	
	\begin{align*}
	&\mathbb{E}[D_{t+1}]=\sum_{j}\biggl[\frac{v_{j}\overline{a_{j}q_{j}}}{\sum v_{i}\overline{a_{i}}}\cdot\frac{\sum_{i\not=j}v_{i}\overline{a_{i}}\overline{q_{i}}+v_{j}(\overline{a_{j}}+1-c_j)\overline{q_{j}}}{\sum_{i\not=j}v_{i}\overline{a_{i}}+v_{j}(\overline{a_{j}}+1-c_j)}\biggr]\\
	&\quad \quad \quad \quad + \biggl[1-\frac{\sum_{i}v_{i}\overline{a_{i}q_{i}}}{\sum_{i}v_{i}\overline{a_{i}}}\biggr]\cdot\frac{\sum_{i}v_{i}\overline{a_{i}q_{i}}}{\sum_{i}v_{i}\overline{a_{i}}}\\
	&\quad \quad \quad \quad =\sum_{j}\biggl[\frac{v_{j}\overline{a_{j}q_{j}}}{\sum v_{i}\overline{a_{i}}}\cdot\frac{\sum_{i}v_{i}\overline{a_{i}q_{i}}+v_{j}(1-c_j)\overline{q_{j}}}{\sum_{i}v_{i}\overline{a_{i}}+v_{j}(1-c_j)}\biggr]\\
	&\quad \quad \quad \quad +\biggl[1-\frac{\sum_{j}v_{j}\overline{a_{j}q_{j}}}{\sum_{i}v_{i}\overline{a_{i}}}\biggr]\cdot\overline{\lambda}
	\end{align*}
	
	\noindent	
	We need to prove that
	\begin{equation}
	\mathbb{E}[D_{t+1}] \geq \mathbb{E}[D_{t}],\label{eq:1}
	\end{equation}
	which is equivalent to show, using Equation \ref{lambdaproof}, that
	\begin{equation*}
	\sum_{j}\biggl[\frac{v_{j}\overline{a_{j}q_{j}}}{\sum v_{i}\overline{a_{i}}}\cdot\frac{\sum_{i}v_{i}\overline{a_{i}q_{i}}+v_{j}(1-c_j)\overline{q_{j}}}{\sum_{i}v_{i}\overline{a_{i}}+v_{j}(1-c_j)}\biggr] +[1-\overline{\lambda}]\cdot\overline{\lambda}\geq\overline{\lambda}
	\end{equation*}
	
	\noindent
	Rearranging the terms, the proof obligation becomes 
	\[
	\frac{1}{\sum_{i}v_{i}\overline{a_{i}}}\sum_{j}\left[\frac{v_{j}^2\overline{a_{j}q_{j}}(1-c_j)}{\sum_{i}v_{i}\overline{a_{i}}+v_{j}(1-c_j)}\left(\overline{q_{j}}-\overline{\lambda}\right)\right] \geq 0
	\]
	or, equivalently,
	\begin{equation}\label{eq:condition1}
	\sum_{j}\left[\frac{v_{j}^2\overline{a_{j}q_{j}}(1-c_j)}{\sum_{i}v_{i}\overline{a_{i}}+v_{j}(1-c_j)}\left(\overline{q_{j}}-\overline{\lambda} \right)\right] \geq 0.
	\end{equation}
	
	\noindent
	Let $k=\max \{ i \in N | (\overline{q_i}-\lambda )\geq 0 \} $, i.e.,
	the largest $k\in N$ such that $q_k\geq \lambda$. By separating the
	sum into positive and negative terms, we obtain
	\begin{align*}
	&\sum_{j}\left[\frac{v_{j}^2\overline{a_{j}q_{j}}(1-c_j)\left(\overline{q_{j}}-\overline{\lambda} \right)}{\sum_{i}v_{i}\overline{a_{i}}+v_{j}(1-c_j)}\right]=S^+ + S^-\quad \text{ where}\\
	&S^+=\sum_{j=1}^k \biggl[\frac{v_{j}\overline{q_{j}}(1-c_j)}{\sum_{i}v_{i}\overline{a_{i}}+v_{j}(1-c_j)} \overline{a_{j}}v_{j}(\overline{q_{j}}-\overline{\lambda})\biggr],\\ &S^-=\sum_{j=k+1}^n\biggl[\frac{v_{j}\overline{q_{j}}(1-c_j)}{\sum_{i}v_{i}\overline{a_{i}}+v_{j}(1-c_j)}\overline{a_{j}}v_{j}(q_{j}-\overline{\lambda})\biggr]
	\end{align*}
	
	\noindent
	By definition of $k$, all the terms in $S^+$ are positive and the
	terms in $S^-$ are negative. Now, by definition of $k$ and $\overline{q_i}=\frac{q_i}{(1-c_i)}$, we have
	\begin{align}\label{q_boundsopt}
	&\forall i \leq k: \; (1-c_i) \leq \frac{q_i}{\overline{\lambda}},\nonumber\\
	&\forall i > k: \; (1-c_i) \geq \frac{q_i}{\overline{\lambda}}.
	\end{align}
	
	\noindent
	We now compute a lower bound for $S^+$ and $S^-$. For $S^+$, using
	Equation \ref{q_boundsopt} for $j\leq k$, we have
	
	\begin{align}\label{s_plus}
	&\frac{v_{j}\overline{q_{j}}(1-c_j)}{\sum_{i}v_{i}\overline{a_{i}}+v_{j}(1-c_j)}\leq \frac{v_{j}q_j}{\sum_{i}v_{i}\overline{a_{i}}+v_{j}\frac{q_j}{\overline{\lambda}}}\nonumber\\
	&\leq \overline{\lambda}\cdot \frac{v_j q_j}{\overline{\lambda}\sum_{i}v_{i}\overline{a_{i}}+v_{j}q_j}\leq \overline{\lambda}\cdot \frac{v_k q_k}{\overline{\lambda}\sum_{i}v_{i}\overline{a_{i}}+v_{k}q_k}
	\end{align}
	
	\noindent
	The last inequality follows by $v_i\geq v_k$ and $q_i\geq q_k$ (using Theorem \ref{monotony_Poly}) and the following property:
	For all $c>0$ and $x \geq y \geq0$, 
	\[
	\frac{x}{c+x}\geq\frac{y}{c+y}\Leftrightarrow(c+y)x\geq(c+x)y\Leftrightarrow cx\geq cy\Leftrightarrow x\geq y.
	\]
	For $S^-$, consider the following expression for $j>k$:
	\[
	\overline{\lambda}(\sum_{i}v_{i}\overline{a_{i}})\underbrace{[v_jq_j-v_kq_k]}_{\geq 0}+ v_kq_k\underbrace{[v_jq_j-\overline{\lambda}v_j(1-c_j)]}_{\geq 0} 
	\]
	Here the first term is greater or equal than zero because $v_j\geq
	v_k$ and $q_j\geq q_k$ using Theorem \ref{monotony_Poly} again. The
	second term is also greater than zero because it can be lower-bounded
	(using Equation \ref{q_boundsopt}) by:
	\[v_jq_j-\overline{\lambda}\frac{v_jq_j}{\overline{\lambda}} =0.\]
	Hence, 
	\begin{align}\label{s_minus}
	&\overline{\lambda}(\sum_{i}v_{i}\overline{a_{i}})[v_jq_j-v_kq_k]+ v_kq_k[v_jq_j-\overline{\lambda}v_j(1-c_j)]\geq 0 \nonumber\\
	&v_jq_j[\overline{\lambda}(\sum_{i}v_{i}\overline{a_{i}}) +v_kq_k] \geq \overline{\lambda}v_kq_k[\sum_{i}v_{i}\overline{a_{i}} + v_j(1-c_j)]\nonumber\\
	&\Leftrightarrow \frac{v_jq_j}{\sum_{i}v_{i}\overline{a_{i}} + v_j(1-c_j)}\geq\frac{\overline{\lambda}v_kq_k}{\overline{\lambda}\sum_{i}v_{i}\overline{a_{i}}+v_kq_k}
	\end{align}
	
	\noindent	
	Putting together Equations \ref{s_plus} and \ref{s_minus} gives us a lower bound to $S^++S^-$:
	\begin{equation}\label{lowerboundS}
	S^++S^-=\frac{\overline{\lambda}v_kq_k}{\overline{\lambda}\sum_{i}v_{i}\overline{a_{i}}+v_kq_k}\cdot \sum\limits_{i=1}^{n}v_i\overline{a_i}(\overline{q_i}-\overline{\lambda})
	\end{equation}
	
	\noindent	
	Now, by definition of $\overline{\lambda}$,
	\[
	\overline{\lambda} = \frac{\sum_{i=1}^n v_{i}\overline{a_{i}}\overline{q_{i}}}{\sum_{i=1}^n  v_{i}\overline{a_{i}}}
	\Leftrightarrow
	\sum_{i=1}^n  v_{i}\overline{a_{i}}(\overline{q_{i}}-\overline{\lambda}) = 0.
	\]
	which implies that
	\[
	\frac{\overline{\lambda}v_kq_k}{\overline{\lambda}\sum_{i}v_{i}\overline{a_{i}}+v_kq_k}\cdot \sum\limits_{i=1}^{n}v_i\overline{a_i}(\overline{q_i}-\overline{\lambda})=0
	\]
	concluding the proof.
\end{proof}

\paragraph{Relationship with the Cascade Model:}Observe that the
quality ranking over the continuation quality orders the products in
decreasing order of $\frac{q_{i}}{1-c_{i}}$ value, which is exactly the   
adjusted ecpm from \cite{Aggarwal_2008,Kempe_2008} with all the
revenues set to $1$. Obviously, the quality ranking (when the continuation probabilities preserve the quality rank in the continuation model), and hence the    
adjusted ecpm ranking, are not the best rankings to show to an incoming
participant (the performance ranking is), but our results show that 
they have nice asymptotic properties. 

\section{Experimental Results}
\label{section:experiments}

This section report computational results to highlight the theoretical
analysis. The computational results use settings that model the
\musiclab{} experiments discussed in
\cite{salganik2006experimental,krumme2012quantifying,PLOSONESI}.  As
mentioned in the introduction, \musiclab{} is a trial-offer market
where participants can try a song and then decide to download it. The
experiments use an agent-based simulation to emulate \musiclab{}.
Each simulation consists of $N$ steps and, at each iteration $t$,
\begin{enumerate}
	\item we simulate selecting a song $i$ according to the
	probabilities $p_i(\sigma,d)$, where $\sigma$ is the ranking
	proposed by the policy under evaluation and $d$ is the social
	influence signal.
	
	\item with probability $q_i$, the sampled song is downloaded, in
	which case the simulator increases the social influence signal
	for song $i$, i.e., $d_{i,t+1} = d_{i,t} + 1$. Otherwise,
	$d_{i,t+1} = d_{i,t}$, and if the continuation model is
	used, the simulation goes back to Step 1 with probability
	$c_i$ and advances to the next step otherwise.
\end{enumerate}

\noindent
Every $T$ iterations, a new list $\sigma$ is computed using one of the    
ranking policies. The experimental setting, which aims at being close
to the \musiclab{} experiments, considers 50 songs and simulations
with 20,000 steps. The songs are displayed in a single column. The
analysis in \cite{krumme2012quantifying} indicated that participants
are more likely to try songs higher in the list. More precisely, the
visibility decreases with the list position, except for a slight
increase at the bottom positions. This paper uses the first setting
for qualities, appeals and visibilities from \cite{PLOSONESI}, where   
the quality and the appeal are chosen independently according to a
Gaussian distribution normalized to fit between $0$ and $1$. In
addition, the experiments consider 12 different continuation
probabilities, varying $\rho$ and the power $r$ as shown in Figure
\ref{cont_p}. The results were obtained by averaging $W=100$
simulations.    

\begin{table}[H]
	\begin{center}
		\begin{footnotesize}
			\begin{tabular}{@{} lcccc @{}}
				\toprule
				Parameters   &   {\sc P-rank} & {\sc Q-rank} & {\sc D-rank} & {\sc R-rank} \\
				\midrule
				$\rho=0.1, r=0$  & 5.3\% & 4.9\% & 5.8\% & 7.5\%\\
				$\rho=0.1, r=0.25$  & 4.1\% & 4.3\% & 4.9\% & 5.2\%\\
				$\rho=0.1, r=1$  & 2.2\% & 2.4\% & 2.6\% & 2.5\%\\
				$\rho=0.1, r=2$  & 1.4\% & 1.4\% & 0.2\% & 1\%\\
				$\rho=0.5, r=0$  & 30.6\% & 31\% & 38.3\% & 51.8\%\\
				$\rho=0.5, r=0.25$  & 24.2\% & 24.6\% & 28.4\% & 33.6\%\\
				$\rho=0.5, r=1$  & 13.4\% & 13.2\% & 14.8\% & 12.2\%\\
				$\rho=0.5, r=2$  & 7.2\% & 7.3\% & 6.2\% & 4.6\%\\
				$\rho=0.9, r=0$  & 67.3\% & 67.7\% & 93.9\% & 143.3\%\\
				$\rho=0.9, r=0.25$  & 51.6\% & 52.1\% & 65.2\% & 79.9\%\\
				$\rho=0.9, r=1$  & 26.6\% & 26.8\% & 28.5\% & 24.2\%\\
				$\rho=0.9, r=2$  & 13.6\% & 13.7\% & 12.2\% & 8.3\%\\
				\bottomrule
			\end{tabular}
		\end{footnotesize}
		\caption{Improvement in Market Efficiency (in percentage) for the Continuation Model.}
		\label{table:First_setting}
	\end{center}
\end{table}

\begin{table}[H]
	\begin{center}
		\begin{footnotesize}
			\begin{tabular}{@{} lcccc @{}}
				\toprule
				Parameters   &   {\sc P-rank} & {\sc Q-rank} & {\sc D-rank} & {\sc R-rank} \\
				\midrule
				$\rho=0.5, r=0.25$  & 13776.1 & 13804.1 & 12000.1 & 9393.8\\
				$\rho=0.5, r=1$  & 12579.0 & 12565.5 & 10643.7 & 7885.0 \\
				$\rho=0.9, r=0.25$  & 16784.7 & 16840.9 & 15435.7 & 12680.8\\
				$\rho=0.9, r=1$  & 14041.4 & 14059.1 & 11926.5 & 8741.7\\
				\bottomrule
			\end{tabular}
		\end{footnotesize}
		\caption{Market Efficiency in the Continuation Model.}
		\label{table:First_setting_totaldownloads}
	\end{center}
\end{table}


Table \ref{table:First_setting} presents results on market efficiency
(i.e., the number of downloads) for the trial and offer market with
continuation. The most interesting message from these results, is the       
observation that the popularity and random rankings improve more than
the performance and quality rankings, unless the quality has less impact ($r=2$, because the continuation is decreasing in $r$) in the continuation. 
This can be explained by    
the fact that the continuation provides a way to correct a potentially  
weak ranking. However, as indicated in Table
\ref{table:First_setting_totaldownloads}, this correction is not
enough to bridge the gap with the performance and quality rankings.

\begin{figure}[H]
	%
	%
	\begin{centering}
		\includegraphics[width=0.45\linewidth]{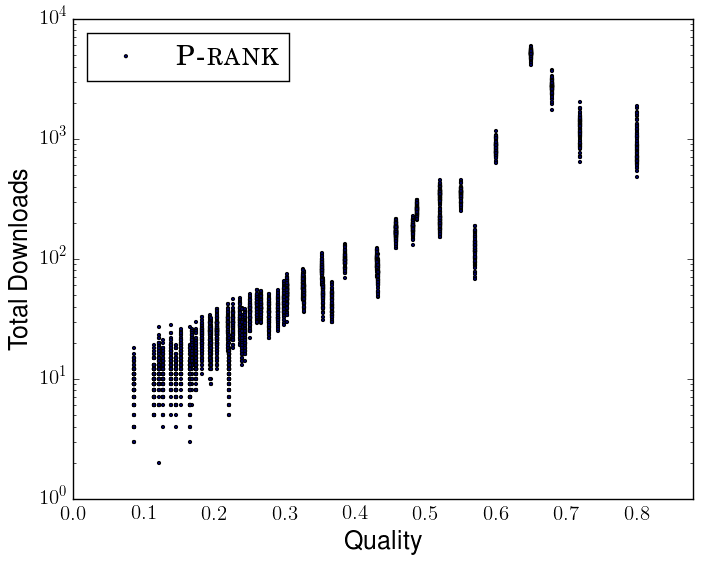}
		\includegraphics[width=0.45\linewidth]{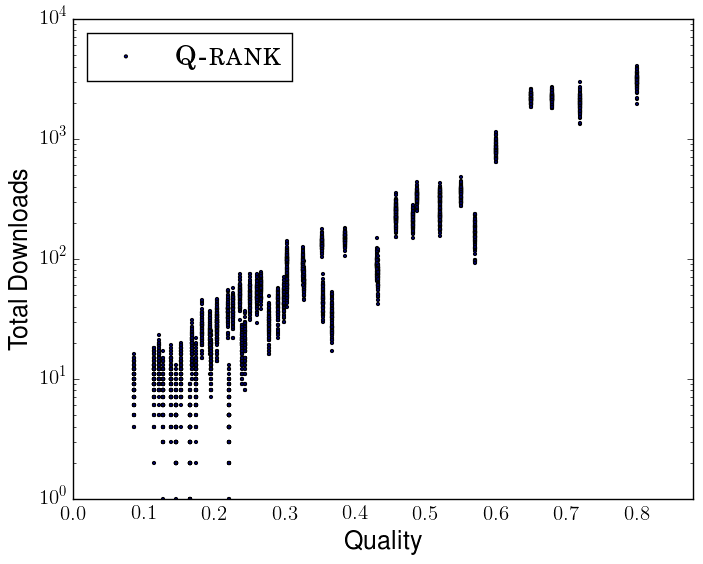} 
		\includegraphics[width=0.45\linewidth]{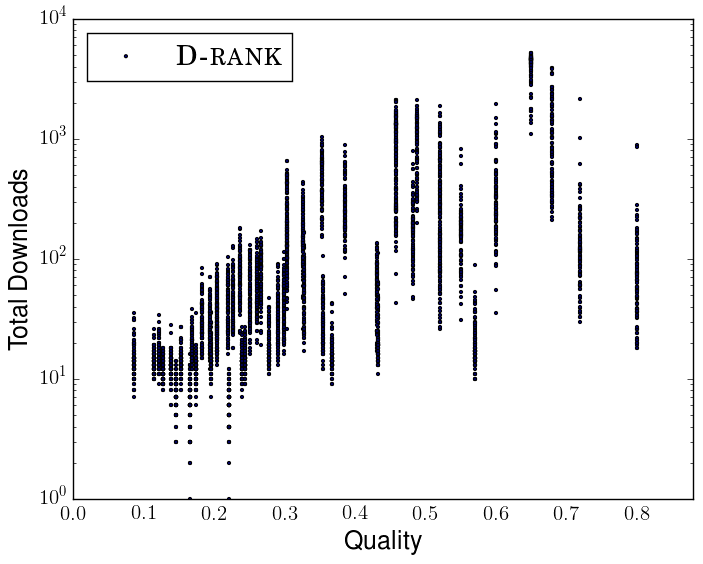}
	\end{centering}
	\centering{}

	\caption{The Distribution of Downloads Versus Song Qualities for $\rho=0.9$, $r=1$. The songs
		on the x-axis are ranked by increasing quality from left to
		right. Each dot is the number of download of a product in
		one of the 100 experiments.}
	\label{fig:q1}
\end{figure}

Figure \ref{fig:q1} depicts experimental results on the predictability
of the market under the continuation model under various ranking
policies. The figure plots the number of
downloads of each song for 100 experiments. In the plots, the songs
are ranked by increasing quality from left to right on the
x-axis. Each dot in the plot shows the number of downloads of a song
in one of the 100 experiments. The results are essentially unchanged    
when moving from the traditional to a continuation multinomial logit
model. The popularity ranking still exhibits significantly more
unpredictability than the performance and quality rankings and the
continuations are not able to compensate for the inherent
unpredictability.
\section{Conclusion and Future Work}

Motivated by applications in online markets, this paper generalizes
the ubiquitous multinomial logit model to a setting that allows market
participants to sample multiple products before deciding whether to
purchase. The paper showed that trial-offer markets with continuation
can be reduced to the original trial-offer model, transferring many
fundamental properties of ranking policies to a more general
setting. In particular, the quality ranking still benefits from
position bias and social influence. Moreover, under a general class of
continuations, the quality ranking is also preserved and the market
reaches the same asymptotic equilibrium. Experimental
results shows that the continuation model compensates for some of the
weaknesses of the popularity ranking by boosting its market
performance more than the quality and performance ranking, unless the
continuation probability depends too strongly on quality. Our current
research aims at generalizing these results further to hierarchical
trial-offer markets.

\newpage

\bibliographystyle{aaai}
\bibliography{model}

\newpage

\appendix
\section*{Appendix}
\label{appendix}
In this section we provide the proofs missing from the main text.

\printproofs
\end{document}